\documentclass[12pt,a4paper,reqno]{amsart}

\usepackage{amsaddr}
\usepackage{amsmath}
\usepackage{amsthm}
\usepackage{mathtools}
\usepackage{hyperref}
\usepackage[nameinlink]{cleveref}
\usepackage{xcolor}
\definecolor{medium-gray}{gray}{0.45} 

\usepackage{algorithm}
\usepackage{algpseudocodex}

\usepackage{multirow}
\usepackage{booktabs}

\usepackage[style=authoryear-comp,backend=biber,giveninits=true,uniquename=init,maxbibnames=10,sorting=nyt,sortcites=no]{biblatex}
\numberwithin{equation}{section}

\newtheorem{theorem}{Theorem}[section]
\newtheorem{lemma}{Lemma}[section]
\newtheorem{corollary}{Corollary}[section]

\newcommand{\wishart}[1]{\mathcal{W}_{#1}}
\newcommand{\invwishart}[1]{\mathcal{W}^{-1}_{#1}}
\newcommand{\cholwishart}[1]{\mathrm{chol}\mathcal{W}_{#1}}
\newcommand{\cholinvwishart}[1]{\mathrm{chol}\mathcal{W}^{-1}_{#1}}
\newcommand{\bigo}{\mathcal{O}}
\newcommand{\cholU}{\operatorname{cholU}}
\renewcommand{\det}[1]{\left| {#1} \right|}
\newcommand{\dmeasure}[1]{\left(\dd{#1}\right)}
\newcommand{\norm}[1]{\left\lVert {#1} \right\rVert}
\newcommand{\dd}[1]{\mathrm{d}{#1}}
\newcommand{\tr}{\operatorname{tr}}
\newcommand{\zeromat}[2]{0_{#1,#2}}

\begin{document}

\title[Efficiently generating inverse-Wishart matrices]{Efficiently generating inverse-Wishart matrices and their Cholesky factors}
\author{Seth D. Axen}
\address{University of Tübingen, Germany}
\email{seth@sethaxen.com}

\begin{abstract}
    This paper presents a new algorithm for generating random inverse-Wishart matrices that directly generates the Cholesky factor of the matrix without computing the factorization.
    Whenever parameterized in terms of a precision matrix $\Omega=\Sigma^{-1}$, or its Cholesky factor, instead of a covariance matrix $\Sigma$, the new algorithm is more efficient than the current standard algorithm.
\end{abstract}

\maketitle

\section{Introduction}

The inverse-Wishart distribution is the conjugate prior for the covariance matrix of a multivariate normal distribution if the mean is known.
As a result, it is often used in Bayesian analyses \parencite{gelman_bayesian_2013,hoff_multivariate_2009}.
It is derived from the Wishart distribution, which is the conjugate prior for the precision matrix of a multivariate normal distribution if the mean is known and is a key object of interest in random matrix theory.

The standard approach for generating a random draw from the inverse-Wishart distribution, termed an inverse-Wishart matrix, first generates a draw from the Wishart distribution, termed a Wishart matrix, and inverts it.
Variants of this approach are implemented in popular Python\footnote{\texttt{scipy.stats.invwishart} (v1.11.3)}, R\footnote{\texttt{LaplacesDemon::rinvwishart} (v16.1.6), \texttt{MCMCpack::riwish} (v1.6-3), and \texttt{mniw::riwish} (v1.0.1)}, and Julia\footnote{\texttt{Distributions.InverseWishart}(v0.25.102)} packages.
We present an alternative algorithm that, for certain parameterizations of the inverse-Wishart distribution, is more efficient than the standard approach.

\subsection*{Key definitions and notation}

Given a positive definite matrix $X$, the upper Cholesky factor of $X$ is the unique upper triangular matrix $U$ with positive diagonal elements such that $X = U^\top U$.
$\det{X}$ indicates the absolute value of the determinant of $X$, while $\norm{X}$ is the Frobenius norm of $X$, i.e. $\norm{X} = \sqrt{\tr(X^\top X)}$.

To perform changes of variables, we will use the exterior product of differentials to derive absolute determinants of Jacobians.
Following \cite{muirhead_aspects_1982}, $\dd{X}$ refers to the matrix of differentials $\dd{x}_{ij}$ of an $m \times m$ positive definite or upper triangular matrix $X$, and $\dmeasure{X}$ refers to the differential form of the Lebesgue measure on the distinct elements of $\dd{X}$, that is,
\[
    \dmeasure{X} \coloneqq \bigwedge_{j=1}^m \bigwedge_{i=1}^j \dd{x}_{ij} = \dd{x}_{11} \wedge \dd{x}_{12} \wedge \dd{x}_{22} \wedge \ldots \wedge \dd{x}_{mm},
\]
where $\wedge$ is the exterior (or wedge) product with the antisymmetric property $\dd{a} \wedge \dd{b} = -(\dd{b} \wedge \dd{a})$ so that $\dd{a} \wedge \dd{a} = 0$.

\subsection*{Change of variables}

Let $X = f(Y)$, where $f$ is a diffeomorphic (smooth, one-to-one function) function.
Then $\dmeasure{X} = \det{J_f} \dmeasure{Y}$, where $J_f$ is the Jacobian of $f$ at $Y$.
If $p_X(X)$ is an absolutely continuous density with respect to $\dmeasure{X}$, then the change of variables formula gives the corresponding density of $Y$ with respect to $\dmeasure{Y}$ as
\[
    p_Y(Y) = p_X(f(Y)) \det{J_f}.
\]

\subsection*{The Wishart distribution}

Let $Z$ be an $m \times n$ matrix whose $i$th column $Y_i \sim \mathcal{N}(0, \Sigma)$, where the scale matrix $\Sigma$ is an $m \times m$ positive definite matrix.
Then the $m \times m$ positive definite matrix $A = Y Y^\top$ is distributed according to the Wishart distribution $\wishart{m}(n, \Sigma)$, where $n$ is termed the degrees of freedom.
This paper considers only the case where $n > m-1$, so that $A$ is nonsingular.
The density of this Wishart distribution with respect to $\dmeasure{A}$ is \parencite[Theorem 3.2.1]{muirhead_aspects_1982}
\begin{equation}\label{wishartdensity}
    p_A(A \mid n, \Sigma) \propto \frac{}{} \exp\left(-\frac{1}{2}\tr(\Sigma^{-1} A)\right) \det{A}^{(n-m-1)/2}.
\end{equation}
The normalization constant is known but unimportant for this paper.

\subsection*{The inverse-Wishart distribution}

If $B=A^{-1}$, then $B$ follows the inverse-Wishart distribution with degrees of freedom $n$ and scale matrix $\Omega = \Sigma^{-1}$, written $\invwishart{m}(n, \Omega)$.
The density of this distribution with respect to $\dmeasure{B}$ is
\begin{equation}\label{invwishartdensity}
    p_B(B \mid n, \Omega) \propto \exp\left(-\frac{1}{2}\tr(\Omega B^{-1})\right) \det{B}^{-(n+m+1)/2}.
\end{equation}

The inverse-Wishart distribution may also be parameterized in terms of $\Sigma$.
While the two parameterizations are equivalent, the choice of parameterization can affect the efficiency of both density computation and sampling algorithms.

\subsection*{Summary of contributions}

\sloppy
This paper's main theoretical contribution is \Cref{thm:invwishartbartlett}, which is applied in \Cref{alg:rinvwishartchol} for efficiently generating a random Cholesky factor of a inverse-Wishart matrix without performing a Cholesky factorization.
\Cref{alg:rinvwishartdirect} uses this algorithm to efficiently generate an inverse-Wishart matrix.
In \Cref{algcomp} we compare the new and standard algorithm and recommend which to use depending on the parameterization of the inverse-Wishart distribution.

\section{The Cholesky parameterizations of the Wishart and inverse-Wishart distributions}

It is frequently more convenient to work with Cholesky factors of Wishart- and inverse-Wishart-distributed matrices than the matrices themselves, since from such factors, the matrix inverse and determinant are easily computed.
For example, if the inverse-Wishart distribution is used to draw random covariance matrices to parameterize a multivariate normal distribution, then computing the density of the multivariate normal would generally require a Cholesky factor.
Thus, it is useful to define the corresponding distributions of Cholesky factors.

\subsection*{The Cholesky-Wishart distribution}

\sloppy
Let $U_A$ and $U_\Sigma$ be the upper Cholesky factors of $A$ and $\Sigma$, respectively.
Then $U_A$ is distributed according to the Cholesky-Wishart distribution $\cholwishart{m}(n, U_\Sigma)$.
\Cref{thm:choleskydetjac} allows us to derive the density of this distribution in \Cref{thm:cholwishartdensity}.

\begin{theorem}[{\cite[Theorem 2.1.9]{muirhead_aspects_1982}}]\label{thm:choleskydetjac}
    Let $X$ be an $m \times m$ positive definite matrix and $T$ be its upper Cholesky factor.
    Then
    \begin{equation}\label{eq:choleskydetjac}
        \dmeasure{X} = 2^m \left(\prod_{j=1}^m t_{jj}^{m+1-j} \right) \dmeasure{T}.
    \end{equation}
\end{theorem}

\begin{theorem}[Cholesky-Wishart density]\label{thm:cholwishartdensity}
    The density of $\cholwishart{m}(n, U_\Sigma)$ with respect to $\dmeasure{U_A} $ is
    \begin{equation}\label{eq:cholwishartdensity}
        p_{U_A}(U_A \mid n, U_\Sigma) \propto \exp\left(-\frac{1}{2}\lVert U_A U_\Sigma^{-1} \rVert^2\right) \prod_{j=1}^m (u_A)_{jj}^{n-j}.
    \end{equation}
\end{theorem}
\begin{proof}
    From \Cref{wishartdensity} and \Cref{eq:choleskydetjac}, and using the change of variables formula,
    \begin{align*}
        p_{U_A}(U_A \mid n, U_\Sigma) & \propto \exp\left(-\frac{1}{2}\tr((U_\Sigma^\top U_\Sigma)^{-1} U_A^\top U_A)\right) \det{U_A}^{n-m-1} \prod_{j=1}^m (u_A)_{jj}^{m+1-j} \\
                                      & = \exp\left(-\frac{1}{2}\tr(U_\Sigma^{-1} U_\Sigma^{-\top} U_A^\top U_A)\right) \prod_{j=1}^m (u_A)_{jj}^{n-j}.
    \end{align*}
    The cyclic property of the trace implies
    \[\tr(U_\Sigma^{-1} U_\Sigma^{-\top} U_A^\top U_A) = \tr((U_A U_\Sigma^{-1})^\top U_A U_\Sigma^{-1}) = \lVert U_A U_\Sigma^{-1} \rVert^2.\]
    Substituting this back into the previous equation completes the proof.
\end{proof}

\subsection*{The Cholesky-inverse-Wishart distribution}

Similarly, if $U_B$ and $U_\Omega$ are the upper Cholesky factors of $B$ and $\Omega$, respectively, then $U_B$ is distributed according to the Cholesky-inverse-Wishart distribution $\cholinvwishart{m}(n, U_\Omega)$.

\begin{theorem}[Cholesky-inverse-Wishart density]
    \sloppy{The density of $\cholinvwishart{m}(n, U_\Omega)$ with respect to the measure $\dmeasure{U_B}$ is}
    \begin{equation}\label{cholinvwishartdensity}
        p_{U_B}(U_B \mid n, U_\Omega) \propto \exp\left(-\frac{1}{2}\lVert U_\Omega U_B^{-1} \rVert^2\right) \prod_{j=1}^m (u_B)_{jj}^{-(n+j)}.
    \end{equation}
\end{theorem}
\begin{proof}
    Starting from \Cref{invwishartdensity} and \Cref{eq:choleskydetjac}, the proof follows the same steps as that of \Cref{thm:cholwishartdensity}.
\end{proof}

One could alternatively parameterize the Cholesky-inverse-Wishart distribution in terms of $U_\Sigma$ rather than $U_\Omega$, though this would increase the cost of computing the density.

\section{Generating random inverse-Wishart matrices}

\subsection{The standard approach}

Random Wishart matrices are usually generated using the so-called Bartlett decomposition of the Cholesky factor of a Wishart matrix into the product of two Cholesky factors \parencite{bartlett_xxtheory_1934}:
\begin{theorem}[{\cite[Theorem 3.2.14]{muirhead_aspects_1982}}]\label{thm:wishartbartlett}
    Let $U_A = Z U_\Sigma$.
    Then
    \begin{equation}
        z_{ij} \sim \begin{cases}\mathcal{N}(0, 1) & i < j\\ \chi_{n+1-j} & i = j\end{cases} \iff U_A \sim \cholwishart{m}(n, U_\Sigma).
    \end{equation}
\end{theorem}
The resulting algorithm from \cite{smith_algorithm_1972} is summarized in \Cref{alg:rwishartchol}.

\begin{algorithm}
    \caption{Generate a Cholesky-Wishart matrix.}\label{alg:rwishartchol}
    \begin{algorithmic}[1]
        \Function{rwishart\_chol}{$m,n,U_\Sigma$} \Comment{$\cholwishart{m}(n, U_\Sigma)$}
        \State $Z \gets \zeromat{m}{m}$
        \For{$j\in 1\ldots m$}
        \For{$i \in 1\ldots j-1$}
        \State Generate $z_{ij} \sim \mathcal{N}(0, 1)$
        \EndFor
        \State Generate $z_{jj} \sim \chi_{n+1-j}$
        \EndFor \Comment{$Z \sim \cholwishart{m}(n, I_m)$}
        \State $U_A \gets Z U_\Sigma$
        \State \textbf{return} $U_A$
        \EndFunction
    \end{algorithmic}
\end{algorithm}

The standard algorithm for generating an inverse-Wishart matrix is to first generate a Wishart matrix using \Cref{alg:rwishartchol} and then compute the matrix inverse.
As noted by \cite{jones_generating_1985}, this is most efficiently computed by inverting the Cholesky factor generated by \Cref{alg:rwishartchol}.
The resulting sampling procedure is summarized in \Cref{alg:rinvwishartindirect}, where $\cholU(X)$ computes the upper Cholesky factor of the positive definite matrix $X$.
It uses \Cref{alg:cholupper} to support all common parameterizations of the inverse-Wishart distribution.

\begin{algorithm}
    \caption{Compute the upper Cholesky factor of a positive definite matrix or its inverse from the matrix or its upper Cholesky factor.}\label{alg:cholupper}
    \begin{algorithmic}[1]
        \Function{cholesky\_upper}{$S,\text{invert},\text{ischolU}$}
        \If{$\text{invert}$}
        \If{$\text{ischolU}$}
        \State $U \gets S$
        \Else
        \State $U \gets \cholU(S)$
        \EndIf
        \State $C \gets U^{-1}$
        \State $P \gets C C^\top$
        \State \textbf{return} $\cholU(P)$
        \Else
        \If{$\text{ischolU}$}
        \State \textbf{return} $S$
        \Else
        \State \textbf{return} $\cholU(S)$
        \EndIf
        \EndIf
        \EndFunction
    \end{algorithmic}
\end{algorithm}

\begin{algorithm}
    \caption{Generate an inverse-Wishart matrix using \Cref{alg:rwishartchol}.}\label{alg:rinvwishartindirect}
    \begin{algorithmic}[1]
        \Function{rinvwishart\_indirect}{$m$,$n$,$S$,$\text{iscov}$,$\text{ischolU}$, $\text{retcholU}$}
        \State $U_\Sigma \gets \textproc{cholesky\_upper}(S, !\mathrm{iscov}, \mathrm{ischolU})$
        \State Generate $U_A \sim \textproc{rwishart\_chol}(m,n,U_\Sigma)$
        \State $V \gets U_A^{-1}$
        \State $B \gets V V^{\top}$
        \If {$\text{retcholU}$}
        \State $T \gets \cholU(B)$
        \State \textbf{return} $T$
        \Else
        \State \textbf{return} $B$
        \EndIf
        \EndFunction
    \end{algorithmic}
\end{algorithm}

\subsection{A new approach}

We present a similar decomposition to Bartlett's for the inverse-Wishart distribution:
\begin{theorem}\label{thm:invwishartbartlett}
    Let $U_B = Z^{-1} U_\Omega$.
    Then
    \begin{equation}
        z_{ij} \sim \begin{cases}\mathcal{N}(0, 1) & i < j\\ \chi_{n-m+j} & i = j\end{cases} \iff U_B \sim \cholinvwishart{m}(n, U_\Omega).
    \end{equation}
\end{theorem}

To prove this result, we first work out the necessary Jacobian determinants in \Cref{thm:triinvjac} and \Cref{lem:trimuldetjac}.

\begin{theorem}\label{thm:triinvjac}
    Let $T$ be an $m \times m$ invertible triangular matrix and $R=T^{-1}$.
    Then $\dmeasure{T} = \det{R}^{-(m+1)} \dmeasure{R}$.
\end{theorem}
\begin{proof}
    We first differentiate the constraint $R T = I_m$:
    \[
        \dd{R} T + R \dd{T} = 0 \implies \dd{T} = -T \dd{R}T
    \]
    Now we inspect the elements of $\dd{T}$:
    \begin{align*}
        -\dd{t}_{ij} & = \sum_{k=1}^m \sum_{l=1}^m t_{ik} \dd{r}_{kl} t_{lj} = \sum_{k=1}^i \sum_{l=k}^j t_{ik} t_{lj} \dd{r}_{kl}                                         \\
        -\dd{t}_{11} & = t_{11} t_{11} \dd{r}_{11}                                                                                                                         \\
        -\dd{t}_{12} & = t_{11} t_{22} \dd{r}_{12} + \textcolor{medium-gray}{t_{11} t_{12} \dd{r}_{11}}                                                                    \\
        -\dd{t}_{22} & = t_{22} t_{22} \dd{r}_{22} + \textcolor{medium-gray}{t_{21} t_{22} \dd{r}_{12} + t_{21} t_{12} \dd{r}_{11}}                                        \\
        \vdots \ \   &                                                                                                                                                     \\
        -\dd{t}_{ij} & = t_{ii} t_{jj}\dd{r}_{ij} + \textcolor{medium-gray}{\text{terms involving } \dd{r}_{kl} \text{ for } l < j \text{ or } (l = j \text{ and } k < i)}
    \end{align*}
    By the antisymmetric property of the exterior product, $\dd{t}_{ij} \wedge \dd{t}_{ij} = 0$, so if we take the exterior product of all expressions for $\dd{t}_{ij}$ given above starting with $\dd{t}_{11}$ and working down, only the leading terms $t_{ii} t_{jj} \dd{r}_{ij}$ remain:
    \begin{align*}
        \dmeasure{T} & = \bigwedge_{j=1}^m \bigwedge_{i=1}^j \dd{t}_{ij} = \bigwedge_{j=1}^m \bigwedge_{i=1}^j t_{ii} t_{jj} \dd{r}_{ij} = \left(\prod_{j=1}^m t_{jj} \prod_{i=1}^j t_{ii} \right) \dmeasure{R} \\
                     & =\left(\prod_{j=1}^m t_{jj}^{m+1}\right) \dmeasure{R} = \det{T}^{m+1} \dmeasure{R} = \det{R}^{-m-1} \dmeasure{R}
    \end{align*}
\end{proof}

\begin{lemma}\label{lem:trimuldetjac}
    Let $Y = C X$, where $Y$, $C$, and $X$ are $m \times m$ triangular matrices, and $C$ is constant.
    Then $\dmeasure{Y} \propto \dmeasure{X}$.
\end{lemma}
\begin{proof}
    Since the operation is linear, its differential is the same, and its Jacobian depends only on the constant matrix $C$.
\end{proof}

\begin{corollary}
    Let $U_B = Z^{-1} U_\Omega$.
    Then \Cref{thm:triinvjac} and \Cref{lem:trimuldetjac} imply
    \begin{equation}\label{eq:trildivdetjac}
        \dmeasure{U_B} \propto \det{Z}^{-m-1} \dmeasure{Z}.
    \end{equation}
\end{corollary}

\begin{proof}[\textbf{Proof of \Cref{thm:invwishartbartlett}}]
    First we note that $(u_\Omega)_{jj} = z_{jj} (u_B)_{jj}$.
    Then we use the change of variables formula with the density of \Cref{cholinvwishartdensity} and Jacobian determinant of \Cref{eq:trildivdetjac} and simplify:
    \begin{align*}
        p_Z(Z \mid n) & = p_{U_B}(Z^{-1} U_\Omega \mid n, U_\Omega) \det{Z}^{-m-1}                                                                         \\
                      & \propto \exp\left(-\frac{1}{2}\lVert Z \rVert^2\right) \det{Z}^{-m-1} \prod_{j=1}^m (u_B)_{jj}^{-(n+j)}                            \\
                      & = \exp\left(-\frac{1}{2}\lVert Z \rVert^2\right) \prod_{j=1}^m z_{jj}^{-m-1} z_{jj}^{n+j} (u_\Omega)_{jj}^{n+j}                    \\
                      & \propto \left(\prod_{j=1}^m z_{jj}^{n-m+j-1} e^{-z_{jj}^2/2}\right) \left(\prod_{j=2}^m \prod_{i=1}^{j-1} e^{-z_{ij}^2/2}\right)   \\
                      & \propto \left(\prod_{j=1}^m \chi_{n-m+j}(z_{jj})\right) \left(\prod_{j=2}^m \prod_{i=1}^{j-1} \mathcal{N}(z_{ij} \mid 0, 1)\right)
    \end{align*}
\end{proof}

\Cref{thm:invwishartbartlett} motivates \Cref{alg:rinvwishartchol} for generating a random inverse-Wishart Cholesky factor, which we then use in \Cref{alg:rinvwishartdirect} for generating random inverse-Wishart matrices.

\begin{algorithm}
    \caption{Generate a Cholesky-inverse-Wishart matrix.}\label{alg:rinvwishartchol}
    \begin{algorithmic}[1]
        \Function{rinvwishart\_chol}{$m,n,U_\Omega$}\Comment{$\cholinvwishart{m}(n, U_\Omega)$}
        \State $Z \gets \zeromat{m}{m}$
        \For{$j\in 1\ldots m$}
        \For{$i\in 1\ldots j-1$}
        \State Generate $z_{ij} \sim \mathcal{N}(0, 1)$
        \EndFor
        \State Generate $z_{jj} \sim \chi_{n-m+j}$
        \EndFor
        \State $C \gets Z^{-1}$ \Comment{$C \sim \cholinvwishart{m}(n, I_m)$}
        \State $U_B \gets C U_\Omega$
        \State \textbf{return} $U_B$
        \EndFunction
    \end{algorithmic}
\end{algorithm}

\begin{algorithm}
    \caption{Generate an inverse-Wishart-distributed matrix.}\label{alg:rinvwishartdirect}
    \begin{algorithmic}[1]
        \Function{rinvwishart\_direct}{$m$,$n$,$S$,$\text{iscov}$,$\text{ischolU}$,$\text{retcholU}$}
        \State $U_\Omega \gets \textproc{cholesky\_upper}(S, \text{iscov}, \text{ischolU})$
        \State Generate $U_B \sim \textproc{rinvwishart\_chol}(m,n,U_\Omega)$
        \If {$\text{retcholU}$}
        \State \textbf{return} $U_B$
        \Else
        \State $B \gets U_B^\top U_B$
        \State \textbf{return} $B$
        \EndIf
        \EndFunction
    \end{algorithmic}
\end{algorithm}

Reference implementations of all algorithms are available at \url{https://github.com/sethaxen/rand_invwishart}.

\section{Comparison of algorithms}\label{algcomp}

\begin{table}[h]
    \centering
    \small 
    \begin{tabular}{lccc|ccc}
        \toprule
        \multirow{2}{*}{\parbox{1.5cm}{Parameter                               \\ ($S$)}} & \multicolumn{3}{c}{\Cref{alg:rinvwishartindirect}} & \multicolumn{3}{c}{\Cref{alg:rinvwishartdirect}} \\
        \cmidrule(lr){2-4} \cmidrule(lr){5-7}
                   & TRTRI  & TRMM   & POTRF     & TRTRI  & TRMM      & POTRF  \\
        \midrule
        $\Sigma$   & \bf{1} & \bf{2} & \bf{1(2)} & 2      & 3(2)      & 2      \\
        $U_\Sigma$ & \bf{1} & \bf{2} & \bf{0(1)} & 2      & 3(2)      & 1      \\
        $\Omega$   & 2      & 3      & 2(3)      & \bf{1} & \bf{2(1)} & \bf{1} \\
        $U_\Omega$ & 2      & 3      & 1(2)      & \bf{1} & \bf{2(1)} & \bf{0} \\
        \bottomrule
    \end{tabular}
    \caption{\small Number of $\bigo(m^3)$ BLAS and LAPACK operations required by each algorithm for each parameterization.
        The corresponding count to generate just the Cholesky factor of the random matrix, when different, is given in parentheses.
        For each parameterization the counts for the most efficient algorithm are bolded.
        \\
        TRMM (BLAS): multiplication by a triangular matrix\\
        TRTRI (LAPACK): inversion of a triangular matrix\\
        POTRF (LAPACK): Cholesky factorization
    }
    \label{tab:alg_comparison}
\end{table}

Both \Cref{alg:rwishartchol} and \Cref{alg:rinvwishartchol} require the same number and type of PRNG calls, so \Cref{alg:rinvwishartindirect} and \Cref{alg:rinvwishartdirect} primarily differ in the number of $\bigo(m^2)$ allocations (which can be eliminated by accepting and reusing storage) and in the number of $\bigo(m^3)$ linear algebraic operations.
\Cref{tab:alg_comparison} gives the number of $\bigo(m^3)$ operations required by each algorithm, depending on the parameterization.

If the parameter is specified as $\Sigma$ or $U_\Sigma$, then \Cref{alg:rinvwishartindirect} is the most efficient, while \Cref{alg:rinvwishartdirect} is more efficient if the parameter is specified as $\Omega$ or $U_\Omega$.
The difference is more pronounced when generating random Cholesky factors, since for the standard algorithm, more work is necessary to compute the factor, while for the new algorithm the Cholesky factor is already available.

We therefore recommend \Cref{alg:rinvwishartindirect} if the inverse-Wishart distribution is parameterized by $\Sigma$ or $U_\Sigma$ and \Cref{alg:rinvwishartdirect} if parameterized by $\Omega$ or $U_\Omega$.

\subsection*{Acknowledgments}

Seth Axen is a member of the Cluster of Excellence Machine Learning: New Perspectives for Science, University of Tübingen and is funded by the Deutsche Forschungsgemeinschaft (DFG, German Research Foundation) under Germany’s Excellence Strategy – EXC number 2064/1 – Project number 390727645.

\printbibliography

\end{document}